\documentclass[conference]{IEEEtran}
\usepackage{psfrag}
\usepackage{graphicx}
\usepackage{fancyhdr}
\usepackage{amsmath,amsthm}
\usepackage{amssymb}
\usepackage[all]{xy}
\usepackage{color}
\usepackage{cite}
\usepackage{graphics}
\usepackage{pstricks}
\usepackage{pstricks-add}
\usepackage{pst-plot}
\usepackage{graphicx}
\newcommand{\eqdef}{\triangleq}

\newcommand{\set}[1]{\mathcal{#1}}

\newcommand{\markov}{\textnormal{\mbox{$\multimap\hspace{-0.73ex}-\hspace{-2ex}-$}}}
\newcommand{\calX}{\mathcal{X}}

\newcommand{\ol}{\overline}

\newcommand{\mc}{\mathcal}

\renewcommand{\d}{\mathrm{d}}

\allowdisplaybreaks[4]

\newtheorem{theorem}{Theorem}

\newtheorem{remark}{Remark}
\newtheorem{definition}{Definition}

\theoremstyle{remark}

\allowdisplaybreaks

\allowdisplaybreaks[3]

\begin{document}

\title{Coordination in State-Dependent Distributed Networks: The Two-Agent Case}

 \author{
 \authorblockN{Benjamin Larrousse}
 \authorblockA{
 Centrale-Supelec\\
 91192 Gif-sur-Yvette, France\\
benjamin.larrousse@lss.supelec.fr}  \and 
 \authorblockN{Samson Lasaulce}
 \authorblockA{
CNRS and Centrale-Supelec\\
91192 Gif-sur-Yvette, France\\
lasaulce@lss.supelec.fr}
\and
 \authorblockN{Mich\`ele~Wigger}
 \authorblockA{
 Telecom ParisTech\\
 75013 Paris, France\\
 michele.wigger@telecom-paristech.fr}
}
\maketitle

\begin{abstract} This paper addresses a coordination problem between two 
agents (Agents~$1$ and $2$) in the presence of a noisy communication channel which depends on an external system state $\{x_{0,t}\}$.  The channel takes as inputs both agents' actions, $\{x_{1,t}\}$ and $\{x_{2,t}\}$ and produces outputs that are observed strictly causally at Agent~$2$ but not at Agent~$1$. The system state is available either causally or non-causally at Agent~$1$ but unknown at Agent~$2$. Necessary and sufficient conditions on a joint distribution $\overline{Q}(x_0,x_1,x_2)$  to be implementable asymptotically (i.e, when the number of taken actions grows large) are provided for both causal and non-causal state information at Agent~$1$.
 
Since the coordination degree between the agents' actions, $x_{1,t}$ and $x_{2,t}$, and the system state $x_{0,t}$ is measured in terms of an average payoff function, feasible payoffs are fully characterized by implementable joint distributions. In this sense, our results allow us e.g., to derive the performance of optimal power control policies on an interference channel and to assess the gain provided by non-causal knowledge of the system state at Agent~$1$.

The derived proofs readily yield new results also for the problem of state-communication under a causality constraint at the decoder.  \end{abstract}

\section{Introduction}
\label{sec:introduction}

Performance characterizations of general distributed networks with agents that observe the system state and the actions of some of the other agents, is a prominent open problem  also studied by related disciplines such as control \cite{basar-book-2013} and game theory \cite{aumann-handbook-1994}. In this paper, we contribute to the solution of a special case of this general problem, by treating it as a coordination problem that can be solved using 
joint-source channel codes. This approach has recently been proposed in \cite{Larrousse-isit2013}, see also \cite{Gossner-2006}, {and is expected to extend to setups with more than two agents and  to different observation structures.}

The technical setup under investigation is as follows. We consider two agents that select their actions repeatedly over $T\geq1$ stages (or time-slots) and that wish to coordinate via their actions in the presence of a random system state. At each stage $t\in\{1,\dots,T\}$, the action of Agent $k \in \{1,2\}$ is $x_{k,t}\in\mathcal{X}_k$, with $|\mathcal{X}_k| < \infty$, and the realization of the random system state is $x_{0,t}\in \mathcal{X}_0$ with $|\mathcal{X}_0|< \infty$. The state sequence $X_{0,1}, \ldots, X_{0,T}$ is given by nature and its components are independent and identically distributed (i.i.d.) according to a distribution~$\rho_0$. 

Suppose that each agent has an individual payoff function   $\omega_k\colon \set{X}_0 \times \set{X}_1\times \set{X}_2 \to \mathbb{R}$, for $k\in\{1,2\}$, that is affected by both agents' actions and the nature state. We are interested in determining the set of feasible expected average payoffs 
\begin{equation}\label{eq:payoff}
\overline{\omega}_k^{(T)}= \mathbb{E}  \left[ \frac{1}{T} \sum_{t=1}^T \omega_k(X_{0,t}, X_{1,t}, X_{2,t})  \right],\; k\in\{1,2\},
\end{equation}
that are reachable by some strategies for the agents. This set of feasible expected average payoffs is fully characterized by the set of  feasible averaged distributions on the triples $\{(X_{0,t}, X_{1,t},X_{2,t})\}_{t=1}^T$. In fact,  denoting by $P_{X_{0,t}X_{1,t}X_{2,t}}$ the joint distribution of the time-$t$ tuple $(X_{0,t}, X_{1,t},X_{2,t})$, we have
\begin{IEEEeqnarray*}{rCl}
\overline{\omega}_k^{(T)} &\eqdef&\frac{1}{T} \sum_{t=1}^T \mathbb{E}\left[ \omega_k(X_{0,t}, X_{1,t}, X_{2,t}) \right]  \nonumber \\
&= & \sum_{\substack{x_0\in\set{X}_0\\x_1\in\set{X}_1\\ x_2\in\set{X}_2} }\omega_k(x_{0}, x_{1}, x_{2})   \frac{1}{T} \sum_{t=1}^T  P_{X_{0,t}X_{1,t} X_{2,t}}(x_{0}, x_{1},  x_{2}). \IEEEeqnarraynumspace
\end{IEEEeqnarray*} 
Our main goal in this paper is to determine the set of averaged distributions  $\frac{1}{T} \sum_{t=1}^T P_{X_{0,t}X_{1,t}X_{2,t}}(x_0, x_{1},x_2)$ that can be induced by the agents' strategies. For simplicity, and in order to obtain closed form expressions, we shall focus on the limit $T\to \infty$. 

%

We consider two kinds of scenarios with two different observation structures. In the first scenario---referred to as \emph{non-causal coding}---Agent~$1$ observes the system states \emph{non-causally}. That means, at each stage $t\in\{1,\ldots, T\}$ it knows the entire state sequence $X_0^T=(X_{0,1},\ldots, X_{0,T})$. In the second scenario----referred to as \emph{causal coding}---Agent $1$ learns the states only \emph{causally}. Thus, here, at each stage $t\in\{1,\ldots, T\}$, Agent~$1$ only knows $X_0^t$. 

In both scenarios, Agent $2$ has no direct access to the state nor to Agent~$1$'s actions. Instead, \emph{after} each stage $t$,  Agent~$2$ observes the output 
$y_t\in \mathcal{Y}$, with $|\mathcal{Y}|<\infty$, of a discrete memoryless multi-access channel  that takes as inputs the two agents' actions and the system state. The multi-access channel is assumed memoryless and of transition law   $\Gamma$:
\begin{IEEEeqnarray}{rCl}\label{eq:DMC}
\lefteqn{\textnormal{Pr}\big[Y_t=y_t| X_0^{t}=x_0^t, X_1^t=x_1^t, X_2^t=x_2^t, Y^{t-1}=y^{t-1} \big]}\qquad  \nonumber \\
 & &= \Gamma(y_t|x_{0,t},x_{1,t},x_{2,t}),\hspace{4cm}
\end{IEEEeqnarray}
where throughout this paper we use the shorthand notations $A^m$ and $a^m$ for the tuples $(A_1,\ldots, A_m)$ and $(a_1,\ldots, a_m)$, when $m$ is a positive integer.

The scenario with non-causal coding was introduced in \cite{lasaulce-erice-2013, larrousse-tit-2015-sub}. Special cases, had previously been considered in \cite{Gossner-2006, Cuff-2011, Larrousse-isit2013,letreust-tit-2014-sub}.  Most prominently, Gossner et al \cite{Gossner-2006} solved the first instance of our problem. They considered the case where Agent $2$ can observe strictly causally the system states and Agent~$1$' actions, which corresponds in our setup to  \eqref{eq:DMC} describing the channel  $y_t=(x_{0,t}, x_{1,t})$. Cuff and Zhao presented an alternative proof \cite{Cuff-2011} of the results in \cite{Gossner-2006} based on more traditional  information-theoretic tools and under the framework of ``coordination via actions". The \emph{noisy} communication channel was introduced by  Larrousse et al. in \cite{Larrousse-isit2013}; their channel however did not depend on the system state nor on Agent~$2$'s actions. This same special case has also been addressed by Le Treust in  \cite{letreust-tit-2014-sub}. 

In the present work, we provide a converse proof for the general scenario with non-causal coding {that} establishes optimality of a scheme proposed in \cite{larrousse-tit-2015-sub}. 
We also solve the scenario with causal coding. Our result shows that in this case the agents' optimal strategies are simple and ignore all communication over the channel.

We exemplify our findings at hand of a power control problem. In particular, for this problem, we quantify the loss incurred in performance when the coding has to be performed causally instead of non-causally. 

At last,  the problem under investigation is linked to the information-theoretic \emph{state-communication problem} \cite{ChoudhuriKimMitra-2013,Sutivong2005}. In fact, the proof techniques derived for the coordination problems, immediately yield new results on state-communication when the decoder is restricted to be causal or strictly causal.

\section{Problem formulation and main results}
\label{sec:pb-formulation-main-result}

As explained previously, the distributed network or system comprises two 
agents. These agents take actions in a repeated manner according to their strategies. Strategies are sequences of functions  defined by:
\begin{align}\label{eq:strategies-I}
\text{case NC:}\quad&\left\{
\begin{array}{ccccc}
\sigma_t^{\mathrm{NC}} & : & \calX_0^T & \rightarrow & \calX_1\\
\tau_t& : &  \mathcal{Y}^{t-1} & \rightarrow & \calX_2
\end{array}
 \right.\\\label{eq:strategies-II}
\text{case C:}\quad&\left\{
\begin{array}{ccccc}
\sigma_t^{\mathrm{C}} & \ \ : & \calX_0^t & \rightarrow & \calX_1\\
\tau_t & \ \ : & \mathcal{Y}^{t-1} & \rightarrow & \calX_2
\end{array}
 \right..
\end{align}
where $\mathrm{NC}$ (resp. $\mathrm{C}$) stands for non-causal (resp. causal) coding and for $\mathrm{c} \in  \{ \mathrm{C},  \mathrm{NC}\}$, the functions $(\sigma_t^{\mathrm{c}})_{1\leq t \leq T}$ (resp. $(\tau_t)_{1\leq t \leq T}$ ) describe the strategies employed by Agent $1$ (resp. $2$). The main problem is to characterize the set of joint probability distributions that can be reached when $T\rightarrow \infty$ which we call the set of implementable distributions according to the terminology of  \cite{Larrousse-isit2013}, \cite{Gossner-2006}. Specifically: 

\begin{definition}[Implementability]\label{def:implementability}
For $\mathrm{c} \in \{\mathrm{C}, \mathrm{NC}\}$, the probability distribution $\ol{Q}(x_0,x_1,x_2)$ is implementable if for every $\epsilon>0$ and every sufficiently large $T$,  there exists a pair of strategies $(\sigma_t^{\mathrm{c}}, \tau_t)_{1\leq t \leq T}$ inducing at each stage~$t$ a joint distribution
\begin{align}
  &\mathrm{P}_{X_{0,t}X_{1,t}X_{2,t}} (x_0,x_1,x_2) \eqdef  
 \mathrm{P}_{X_{1,t}X_{2,t}|X_{0,t}}(x_1,x_2|x_0) \rho_0(x_0)\label{eq:form_induced_distribution}
\end{align}
such that for all $(x_0,x_1,x_2)\in \calX_0 \times \calX_1 \times \calX_2$:
\begin{align}
 &\bigg| \frac{1}{T} \sum_{t=1}^{T}   \mathrm{P}_{X_{0,t}X_{1,t}X_{2,t}}(x_0,x_1,x_2) - \ol{Q}(x_0,x_1,x_2)\bigg| \leq \epsilon. \label{eq:implementable}
\end{align}
\end{definition}

We now characterize the set of implementable probability distributions both for causal and non-causal coding. 

\begin{theorem}[Non-causal coding] \label{thm:non-causal} Let $\mathrm{c}=\mathrm{NC}$. Consider a joint  probability distribution $\overline{Q}$ such that $\sum_{x_1,x_2} \ol{Q}(x_0,x_1,x_2) = \rho_0(x_0)$. The distribution $\overline{Q}$ is implementable if and only if it satisfies the following condition{\footnote{{The notation $I_Q(A;B)$ indicates that the mutual information should be computed with respect to the probability distribution $Q$.}}}
\begin{equation}\label{eq:MI_cond}
{I_Q(X_0;X_2) \leq I_Q(V;Y|X_2)  - I_Q(V;X_0|X_2)}
\end{equation}
with $Q(x_0,x_1,x_2,y,v) = \ol{Q}(x_0,x_1,x_2 )\Gamma(y|x_0,x_1,x_2)$ $\times \mathrm{P}_{V|X_0 X_1 X_2}(v|x_0,x_1,x_2)$, and $V$ being an auxiliary random variable which alphabet cardinality can be restricted as $|\mathcal{V}| \leq |\mathcal{X}_0|\cdot|\mathcal{X}_1|\cdot|\mathcal{X}_2|$.
\end{theorem}
\begin{proof}
See Section~\ref{sec:proof_noncausal}.
\end{proof}

\begin{theorem}[Causal coding]\label{thm:causal} Let $\mathrm{c}=\mathrm{C}$. The set of implementable distributions is given by the set of distributions under the form 
\begin{equation}{Q}(x_0,x_1,x_2) = \mathrm{P}_{X_1|X_0X_2}(x_1|x_0,x_2) \mathrm{P}_{X_2}(x_2) \rho_0(x_0).\end{equation}
\end{theorem}

\begin{proof}
See Section~\ref{sec:proof_causal}.
\end{proof}

Note that no information constraint appears in this second theorem. This is related to the fact that in the case of causal coding  no benefit can be obtained by communicating over the channel:  Agent~$2$ can simply ignore the channel outputs. 
In particular, when the two agents are  interested in maximizing a common payoff function $w(x_0, x_1,x_2)$ possible strategies are as follows:  Agent $2$ chooses a constant action $x_{2,t}=x_{2}$, and for each stage~$t$  Agent~$1$ picks an action $x_{1,t}^\star\in\mathcal{X}_1$ in function of this $x_2$ and of the nature state $x_{0,t}$ so as to maximize the payoff function on the current stage: 
\begin{equation}
x_{1,t}^\star{\in}  \arg \max_{\tilde{x}_1\in\mathcal{X}_1} w(x_{0,t}, \tilde{x}_{1}, x_2).
\end{equation} 
This strategy is referred to as a semi-coordinated policy in the context of coded power control introduced in \cite{Larrousse-isit2013}. 


\section{Application: Power Control}\label{sec:numerical-example}

We exemplify the above two theorems at hand of a power control problem. In particular, we wish to illustrate the loss in performance incurred when the coding is only causal instead of non-causal. An interference channel with two transmitter-receiver pairs is considered. Transmissions are assumed to be time-slotted and synchronized. For $k\in\{1,2\}$ and ``$j=-k$'' ($-k$ stands for the transmitter other than $k$), the signal-to-interference plus noise ratio (SINR) at Receiver $k$ at a given stage writes as $\mathrm{SINR}_k= \frac{g_{kk} x_k }{\sigma^2 + g_{jk} x_{j}}$
where $x_k \in \mathcal{X}_k = \left\{0, P_{\max}\right\}$ is the power level chosen by Agent or Transmitter $k$, $g_{kj}$ represents the channel gain of link $kj$, and $\sigma^2$ the noise variance. We assume that: $g_{kj} \in \{g_{\min}, g_{\max}\}$ is Bernouilli distributed $g_{kj} \sim \mc{B}(p_{kj})$ with $\mathrm{P}(g_{kj} = g_{\min}) = p_{kj}$; the global channel state is thus given by $x_0=(g_{11},g_{12},g_{21},g_{22})$. We define $\text{SNR(dB)} = 10\log_{10}\frac{P_{\max}}{\sigma^2}$ and set $g_{\min} = 0.1$, $g_{\max}=2$, $\sigma^2=1$, and $(p_{11},p_{12},p_{21},p_{22}) = (0.5,0.1,0.1,0.5)$. The considered common payoff function is $w(x_0,x_1,x_2) = \sum_{k=1}^2 \log_2(1+\mathrm{SINR}_k)$ and the signal $Y$ observed by Agent/Transmitter $2$ is assumed to be the output of a binary symmetric channel with transition probability $e=0.05$. 

Fig.~\ref{Fig:pc-causal} represents the maximum expected sum-rate against SNR in dB for {our two scenarios with causal and non-causal coding at Transmitter/Agent~$1$. (For practical reasons we restrict to $|\set{V}|=10$.) These two scenarios are compared to a scenario with \emph{costless communication}  {where Agent/Transmitter~$2$ observes $x_{1}^T$ non-causally} and thus the maximum of $w$ can be reached at any stage, and to a scenario where the two agents don't coordinate but simply transmit at full power throughout.}

\begin{figure}[htbp]
  \includegraphics[width=0.53\textwidth]{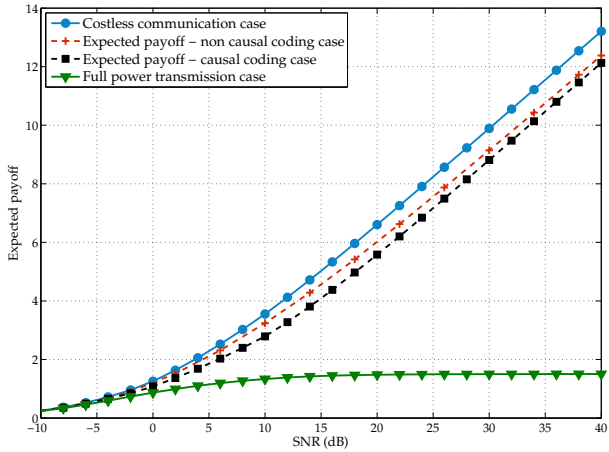}
\caption{Expected payoff against SNR(dB). One message from the figure is that good coordinated power control policies may perform quite close to the maximum sum-payoff for certain standard payoff functions; here the chosen sum-payoff is the sum-rate. Another message is that, for certain standard payoff functions designing non-causal power control policies may not bring very significant performance gains over causal power control policies.}
\label{Fig:pc-causal}
\end{figure}


\section{Related Results on State-Communication}
\label{sec:state-amplification}

Consider again the setup of Sections~\ref{sec:introduction} and \ref{sec:pb-formulation-main-result}, {but---in line of previous works on state-communication \cite{ChoudhuriKimMitra-2013,Sutivong2005}---specialize the channel in \eqref{eq:DMC} to a} state-dependent discrete memoryless channel (DMC) with state $x_{0,t}$ and single input $x_{1,t}$:\footnote{Traditionally, in state-communication the channel inputs are denoted $x_1,\ldots, x_t$, the state symbols $s_1,\ldots, s_T$ and the reconstructed symbols at Agent~$2$ $\hat{s}_1, \ldots, \hat{s}_T$. For coherence, here we keep the notations introduced in the first part of the paper.}\footnote{{Our results readily extend also to the more general channel in \eqref{eq:DMC}.}}
\begin{equation}\
\Gamma(y_t|x_{0,t}, x_{1,t}, x_{2,t}) = \Gamma(y_t|x_{0,t}, x_{1,t}). 
\label{eq:DMC2}
\end{equation}
{Initially,} we consider non-causal coding functions and causal or strictly causal decoding functions: 
\begin{align}\label{eq:strategies-2}
\text{case NC-enc/C-dec:}\quad&\left\{
\begin{array}{ccccc}
\sigma_t^{\mathrm{NC}} & : & \calX_0^T & \rightarrow & \calX_1\\
\tau_t^{\mathrm{C}}& : &  \mathcal{Y}^{t} & \rightarrow & \calX_2
\end{array}
 \right.\\\label{eq:strategies-II-bis}
\text{case NC-enc/SC-dec:}\quad&\left\{
\begin{array}{ccccc}
\sigma_t^{\mathrm{NC}} & : & \calX_0^T & \rightarrow & \calX_1\\
\tau_t^{\mathrm{SC}} & : & \mathcal{Y}^{t-1} & \rightarrow & \calX_2
\end{array}
 \right.
\end{align}
where $\mathrm{C}$ (resp. $\mathrm{SC}$) stands for causal (resp. strictly causal) decoding.

The goal here is that Agent~$2$ can produce a reconstruction sequence $x_2^T$ that matches the state sequence $x_{0}^T$ up to an allowed distortion.
In particular, for $d\in\{\mathrm{C}, \mathrm{SC}\}$, Distortion $D\geq 0$ is said achievable under a given bounded \emph{single-letter distortion function} 
\begin{equation}
\delta \colon \set{X}_0 \times \set{X}_2 \to {[0, d_\text{max} ],}
\end{equation}
 if for every $\epsilon>0$ and sufficiently large blocklengths $T$ it is possible to find encoding and decoding functions $\{\sigma_t^{\mathrm{NC}}\}_{t=1}^T$ and $\{\tau_t^{d}\}_{t=1}^T$ such that Agent~$2$'s reconstructed sequence {$X_{2}^T$} satisfies
\begin{equation}
{ \mathbb{E}\left[ \frac{1}{T} \sum_{t=1}^T \delta(X_{0,t}, X_{2,t})\right] \leq D+ \epsilon.}
\end{equation}

\begin{theorem}[Non-causal coding]\label{thm1}  
Let $d=\mathrm{C}$, i.e., decoding is causal.  Distortion $D$ is achievable if and only if
{
\begin{equation}\label{eq:distor_causal}
\mathbb{E} [\delta(X_0, g(U,Y))] \leq D,
\end{equation}
for some function  $g\colon \set{U}\times \set{Y} \to \set{X}_2$ and some joint law $Q_{X_0X_1YUV}(x_0,x_1,y,u,v)$ that factorizes as 
\begin{IEEEeqnarray}{rCl}\label{eq:cons11}
\rho_0(x_0)  P_{UVX_1|X_0}(u,v,x_1|x_0)  \Gamma(y|x_0,x_1); \quad
\end{IEEEeqnarray} 
and satisfies
\begin{equation}\label{eq:cons21}
I_Q(U;X_0) \leq I_Q(V;Y|U)- I_Q(V;X_0|U).
\end{equation}}
%
%
%
The input $X_1$ can be restricted to be a function of $(U,V, X_0)$.
\end{theorem}
\begin{proof}
Omitted for brevity.
\end{proof}
\begin{remark}
Theorem~\ref{thm1} remains valid for $d=\mathrm{SC}$, i.e., when decoding is strictly causal, if \eqref{eq:distor_causal} is replaced by
\begin{equation}\label{eq:distor_sccausal}
\mathbb{E} [\delta(X_0, g(Y))] \leq D.
\end{equation}
In this case, one can restrict to the choice $U=g(Y)=X_2$.
\end{remark}
%

Assume now that encoding is causal. The setup is as described above, except that 
the encoding functions  in \eqref{eq:strategies-2} have to be replaced by functions of the form
\begin{align}\label{eq:strategies-4}
\sigma_t^{\mathrm{C}}  :  \calX_0^t  \rightarrow  \calX_1.
 \end{align}

\begin{theorem}[Causal coding]\label{thm5}Let $d=\mathrm{C}$, i.e., decoding is causal. 
Distortion $D$ is achievable if and only if
\begin{equation}
\mathbb{E} [ \delta(X_0, g(Y))] \leq D,
\end{equation}
{for some function $g\colon \set{Y}\to \set{X}_2$
and a joint distribution $P_{X_0X_1Y}(x_0,x_1,y)$ that factorizes as }
\begin{IEEEeqnarray}{rCl}\label{eq:jointlaw}
\rho_0(x_0)  P_{X_1|X_0}(x_1|x_0)P_{Y|X_0X_1}(y|x_0,x_1). \quad
\end{IEEEeqnarray} 

 Let $d=\mathrm{SC}$, i.e., decoding is strictly causal. 
Distortion $D$ is achievable if and only if there exists a  constant value $x_2\in \set{X}_2$ so that 
\begin{equation}
\mathbb{E} [ \delta(X_0,x_2)] \leq D. 
\end{equation}
\end{theorem}
\begin{proof}
Omitted for brevity.
\end{proof}
\begin{remark}
In combination with previous results on state-communication \cite{ChoudhuriKimMitra-2013,Sutivong2005}, our results provide the following insights. When the \emph{decoder} is non-causal,  Wyner-Ziv coding has to be used  to compress the state. This is not possible anymore when the decoder is only causal, where  standard compression suffices.  When the \emph{encoder} is non-causal, then Gel'fand-Pinsker coding should be used to communicate over the channel. When the encoder is only causal, this is not possible anymore and the less powerful Shannon strategies suffice. When the encoding and the decoding are causal or strictly causal, then no coding is needed anymore; simple symbol-by symbol strategies at the transmitter (Agent~$1$) and the receiver (Agent~$2$) are sufficient.
\end{remark}

\section{Proof of Theorem~\ref{thm:non-causal}}\label{sec:proof_noncausal}
The proof of Theorem \ref{thm:non-causal} can be divided into three parts: the direct part, which is established in \cite{larrousse-tit-2015-sub} and omitted for brevity;  the bound on the cardinality of the auxiliary alphabet $|\mathcal{V}|$which is also omitted; and the converse, which shows  optimality of the coding scheme in \cite{larrousse-tit-2015-sub}
and is proved in the following. 

\textit{Converse.} 
Let $\overline{Q}$ be an implementable distribution, and fix an arbitrary $\epsilon>0$.

By definition,  there must exist a sufficiently large blocklength $T$ and strategies $\{\sigma_{t}^{\mathrm{NC}}\}_{t=1}^T$ and $\{\tau_{t}\}_{t=1}^T$ such that for  $t\in\{1,\ldots, T\}$, the tuple  $(X_{0,t}, X_{1,t}, X_{2,t})$ induced by these strategies has a joint law $ P_{X_{0,t}X_{1,t}X_{2,t}}$ that satisfies 
\begin{equation}\label{eq:close2}
\bigg| \frac{1}{T}\sum_{t=1}^T P_{X_{0,t}X_{1,t}X_{2,t}}(x_0, x_1,x_2) - \overline{Q}(x_0,x_1,x_2)\bigg|<\epsilon.
\end{equation}
 

For each $t$, let $Y_t$ denote the output of the channel for state $X_{0,t}$ and inputs $X_{1,t}$ and $X_{2,t}$. For each positive integer $m$, we use the shorthand notation  $A_{m}^{T}$ to denote the random tuple $(A_{m}, \ldots, A_T)$.  Let $Z$ be a random variable that is uniformly distributed over $\{1,\ldots, T\}$ independent of $X_0^T, X_1^T, X_2^T,Y^T$, and define for each $t$, 
\begin{IEEEeqnarray}{rCl}\label{eq:Vt}
V_t & \eqdef & (X_{0,t+1}^T, Y^{t-1}).
\end{IEEEeqnarray}
Finally, let 
$V\eqdef (V_Z,Z)$, $X_1\eqdef X_{1,Z}$, $Y\eqdef Y_Z$, $X_0\eqdef X_{0,Z}$, and $X_{2}\eqdef X_{2,Z}$, and 
denote the probability distribution of the tuple $(V,X_0,X_1,X_2, Y)$ by $Q_{VX_0X_1X_2Y}$. 
Notice that this law factorizes as
\begin{align}
\label{eq:channel_factor}\lefteqn{
  Q_{VX_0X_1X_2Y}(v,x_0,x_1,x_2,y) } \quad \nonumber \\
&  = \rho_0(x_0) Q_{VX_1X_2|X_0}(v,x_1,x_2|x_0)\Gamma(y|x_0,x_1,x_2),\end{align}
with a marginal law  satisfying
\begin{IEEEeqnarray}{rCl}\label{eq:average_over_time2}
\lefteqn{
 \sum_{v, y}Q_{VX_0X_1X_2Y}(v,x_0,x_1,x_2, y) }\qquad  \nonumber \\
  &= & \frac{1}{T}\sum_{t=1}^T P_{X_{0,t}X_{1,t}X_{2,t}}(x_0, x_1,x_2).
\end{IEEEeqnarray}
The Markov chain $Y - (X_0,X_1,X_2) - V$ in \eqref{eq:channel_factor} holds because by the memorylessness of the channel~\eqref{eq:DMC}, $Y_t - (X_{0,t}, X_{1,t},  X_{2,t}) - (X_0^T, X_1^T,  X_2^T, Y^{t-1}, Y_{t+1}^T)$ forms a Markov chain for any $t\in\{1,\ldots, T\}$, and because the time-sharing random-variable $Z$ is independent of $(X_0^T, X_1^T, X_2^T, Y^T)$.

\allowdisplaybreaks[4]
We continue with the following sequence of equalities:
\begin{IEEEeqnarray}{rCl}
\lefteqn{\frac{1}{T} I(X_0^T;Y^T)}\nonumber \\
 & \stackrel{(a)}{=}& \frac{1}{T} \sum_{t=1}^T I(X_{0,t};Y^T|X_{0,t+1}^T) \nonumber\\
& \stackrel{(b)}{=} &  \frac{1}{T} \sum_{t=1}^T I(X_{0,t};Y^T, X_{0,t+1}^T) \nonumber\\
& \stackrel{(c)}{=}& \frac{1}{T} \sum_{t=1}^T  \big[ I(X_{0,t};Y^t, X_{0,t+1}^T) + I(X_{0,t};Y_{t+1}^T| Y^t, X_{0,t+1}^T)\big]\nonumber\\
& \stackrel{(d)}{=}& \frac{1}{T} \sum_{t=1}^T  \big[ I(X_{0,t};Y^t, X_{0,t+1}^T, X_{2,t}) \nonumber \\
& & \qquad \qquad + I(X_{0,t};Y_{t+1}^T| Y^t, X_{0,t+1}^T)\big]\nonumber\\
&\stackrel{(e)}{=} & \frac{1}{T} \sum_{t=1}^T  \big[  I(X_{0,t};Y_t,  V_t, X_{2,t}) + 
I(X_{0,t};Y_{t+1}^T| Y_t, V_t)\big],\label{eq:up1}\nonumber \\
& = & I(X_{0,Z}, V_Z; Y_Z|X_{2,Z}, Z) +  I(X_{0,Z};Y_{Z+1}^T| Y_Z,  V_Z, Z) \nonumber \\
& = & I(X_0, V;Y|X_{2}) + I(X_{0,Z};Y_{Z+1}^T| Y_Z,  V_Z, Z), \label{eq:up2}
\end{IEEEeqnarray}
where $(a)$ follows from the chain rule of mutual information; $(b)$ by the i.i.d-ness of the state sequence $(X_{0,1},\dots,X_{0,T})$; $(c)$ by the chain rule of mutual information; $(d)$ because $X_{2,t}$ is computed as a function of $Y^{t-1}$; $(e)$ by \eqref{eq:Vt}; {and the last two equalities by the definitions of $(Z, V_{Z}, X_{0,Z},X_{2,Z}, Y_Z)$ and $(V, X_0,X_2, Y)$ and the independence of $Z$ and $X_{0,Z}$.}

On the other hand, 
\begin{IEEEeqnarray}{rCl}
\lefteqn{\frac{1}{T} I(X_0^T;Y^T)}\nonumber \\
& \stackrel{(f)}{=} & \frac{1}{T} \sum_{t=1}^T \Big[ I(X_{0,t};Y^T|X_{0,t+1}^T)  + I(Y^t;X_{0,t+1}^T) \nonumber \\
&& \phantom{=========} - I(Y^{t-1};X_{0,t}^T) \Big]\nonumber \\
& \stackrel{(g)}{=} & \frac{1}{T} \sum_{t=1}^T \Big[ I(X_{0,t}^T;Y^t)  + I(X_{0,t}; Y_{t+1}^T|Y^{t}X_{0,t+1}^T) \nonumber \\
&& \phantom{=========}  - I(Y^{t-1};X_{0,t}^T) \Big]\nonumber \\
& \stackrel{(h)}{=} & \frac{1}{T} \sum_{t=1}^T \Big[ I(X_{0,t}^T;Y_t |Y^{t-1})  + I(X_{0,t}; Y_{t+1}^T|Y^{t}X_{0,t+1}^T)  \Big]\nonumber \\
& \stackrel{(i)}{=}  & \frac{1}{T} \sum_{t=1}^T \Big[ I(X_{0,t}^T;Y_t |Y^{t-1}, X_{2,t})  + I(X_{0,t}; Y_{t+1}^T|Y^{t}X_{0,t+1}^T)  \Big]\nonumber \\
& \stackrel{(j)}{\leq} &  \frac{1}{T} \sum_{t=1}^T \Big[ I(X_{0,t}^T , Y^{t-1};Y_t |X_{2,t})  + I(X_{0,t}; Y_{t+1}^T|Y^{t} X_{0,t+1}^T)  \Big]\nonumber \\
& \stackrel{(k)}{=} & \frac{1}{T} \sum_{t=1}^T \Big[ I(X_{0,t}, V_t; Y_t|X_{2,t}) +  I(X_{0,t};Y_{t+1}^T| Y_t,  V_t)  \Big] \nonumber \label{eq:up3}\nonumber \\
& \stackrel{(\ell)}{=}& I(X_{0,Z}, V_Z; Y_Z|X_{2,Z}, Z) +  I(X_{0,Z};Y_{Z+1}^T| Y_Z,  V_Z, Z) \nonumber \\
&\stackrel{(m)}{\leq} & I(X_0, V;Y|X_{2}) + I(X_{0,Z};Y_{Z+1}^T| Y_Z,  V_Z, Z), \label{eq:up4}
\end{IEEEeqnarray}
where $(f)$ follows from the chain rule of mutual information and from the {Csisz\'{a}r-Kramer} telescoping identity \cite{kramer-itsl-2011}; $(g)$ and $(h)$ {follows} by the chain rule of mutual information;  $(i)$ because $X_{2,t}$ is computed as a function of $Y^{t-1}$; $(j)$ follows because conditioning cannot increase entropy;  $(k)$  by \eqref{eq:Vt};  {$(\ell)$ by the definitions of $(Z, V_{Z}, X_{0,Z},X_{2,Z}, Y_Z)$; and $(m)$ by the definitions of  $(V, X_0,X_2, Y)$ and the independence of $Z$ and $X_{0,Z}$.}
Combining \eqref{eq:up2} and \eqref{eq:up4}, we obtain
\begin{IEEEeqnarray}{rCl}
I(X_0;Y, V, X_2) \leq  I(X_0, V;Y|X_2),
\end{IEEEeqnarray}
which by  chain rule of mutual information is equivalent to 
\begin{IEEEeqnarray}{rCl}\label{eq:const}
I(X_0;X_2) &\leq& I(X_0, V;Y|X_2)- I(X_0; Y,V|X_2)\nonumber \\
&= & I(V;Y|X_2) - I(V;X_0|X_2). 
\end{IEEEeqnarray}

By   \eqref{eq:channel_factor}, \eqref{eq:average_over_time2}, and \eqref{eq:const}, we conclude that   the joint law $ \frac{1}{T}\sum_{t=1}^T P_{X_{0,t}X_{1,t}X_{2,t}}(x_0, x_1,x_2)$ satisfies the conditions on implementable distributions that we stated in the theorem. In view of~\eqref{eq:close2}, since $\epsilon>0$ can be arbitrary small, and by continuity of mutual information, then also the law $\ol{Q}(x_0, x_1, x_2)$ must satisfy these conditions.

%

\section{Proof of Theorem~\ref{thm:causal}}\label{sec:proof_causal}

\textit{Converse.} Let $\overline{Q}$ be an implementable distribution. Fix $\epsilon>0$ and sufficiently large $T$. By definition, there must exist strategies such that for each $t\in\{1,\ldots, T\}$, the triple  $(X_{0,t}, X_{1,t}, X_{2,t})$ has a joint law $ P_{X_{0,t}X_{1,t}X_{2,t}}$ that satisfies 
\begin{equation}\label{eq:close}
\bigg| \frac{1}{T}\sum_{t=1}^T P_{X_{0,t}X_{1,t}X_{2,t}}(x_0, x_1,x_2) - \overline{Q}(x_0,x_1,x_2)\bigg|<\epsilon.
\end{equation}

Let $Z$ be a random variable that is uniformly distributed over $\{1,\ldots, T\}$ and independent of $X_0^T,  X_1^T, X_2^T,Y^T$. 
Further, define  $X_0\eqdef X_{0,Z}$, $X_1\eqdef X_{1,Z}$, $X_{2}\eqdef X_{2,Z}$,  $Y\eqdef Y_Z$.

Denoting the probability mass function of the triple $(X_0,X_1,X_2)$ by $Q_{X_0X_1X_2}$, by the definitions above,
\begin{equation}\label{eq:average_over_time}
 Q_{X_0,X_1,X_2}(x_0,x_1,x_2) =  \frac{1}{T}\sum_{t=1}^T P_{X_{0,t}X_{1,t}X_{2,t}}(x_0, x_1,x_2).
\end{equation}
We will prove that the law $Q_{X_0X_1X_2}$ factorizes as 
\begin{IEEEeqnarray}{rCl}\label{eq:factors}
\rho_0(x_0) Q_{X_2}(x_2) Q_{X_1|X_2,X_0}(x_1|x_2,x_0).
\end{IEEEeqnarray}
By \eqref{eq:average_over_time}, by continuity, and by \eqref{eq:close}, this will imply that also $\ol{Q}$ factorizes in this way, and thus conclude the proof.  

To prove \eqref{eq:factors}, we first notice that for any $t\in\{1,\ldots, T\}$, by the causality of the decoding, $X_{2,t}$ depends only on $Y^{t-1}$. By the causality of the encoding this latter is independent of $X_{0,t}$.  Thus, $X_{2,Z}\markov Z\markov X_{0,Z}$ forms  a Markov chain. Since $Z$ and $X_{0,Z}=X_0$ are independent, $X_2=X_{2,Z}$ also needs to be independent of $X_0$.
These observations combine to establish that the joint law of $(X_0, X_1, X_2)$ has to factorize as in \eqref{eq:factors}.

\textit{Achievability:} Consider a joint distribution $\overline{Q}(x_0, x_1, x_2)$ that factorizes as $\overline{Q}(x_0,x_1,x_2) =  \rho_0(x_0)\mathrm{P}_{X_2}(x_2) \times \mathrm{P}_{X_1|X_0X_2}(x_1|x_0,x_2)$.
Fix  small $ \epsilon_2>\epsilon_1 >0$ and an arbitrary blocklength $T$. Then, pick a $T$-length {sequence} $x_{2,1}, \ldots, x_{2,T}$ that lies in the typical set $\set{T}_{\epsilon_1}^{(T)}(P_{X_2})$; see  \cite[p.~25]{elgamalkim-2011} for a definition of this typical set.

The two agents produce the following actions. At stage $t\in\{1,\ldots, T\}$, Agent~2 produces $x_{2,t}$. 
Agent~1 produces the random action $X_{1,t}$ that it draws according to  the conditional law $P_{X_1|X_0X_2}( \cdot| x_{0,t}, x_{2,t})$.

We analyze the proposed strategies. 
Define the event:
\begin{align}
E^{(T)} \eqdef \Big\{\big(X_0^T,X_1^T,x_2^T\big)\notin \set{T}_{\epsilon_2}^{(T)}(\overline Q)\Big\}.
\end{align}
By the weak law of large numbers, and the conditional typicality lemma \cite[p.~27]{elgamalkim-2011},
\begin{equation}\label{eq:E0_limit}
    \lim_{T\rightarrow\infty}\mathrm{P}\big(E^{(T)}\big)=0. 
\end{equation}
Since $\epsilon_2>0$ can be chosen arbitrarily small, by Proposition~5 in \cite{larrousse-tit-2015-sub}, this establishes the desired achievability result.

\appendices
\section{Proof of Cardinality Bound} \label{sec:cardinality_bound}

Let us prove that in Theorem \ref{thm:non-causal} it suffices to choose $V$ of cardinality 
\begin{equation}\label{eq:cardinality_bound}
|\mathcal{V}| \leq |\mathcal{X}_0|\cdot|\mathcal{X}_1|\cdot|\mathcal{X}_2|.
\end{equation}
Let $\mathcal{P}$ denote the set of pmfs over $\mathcal{X}_0 \times \mathcal{X}_1\times\mathcal{X}_2$. 
For each triple $(x_0, x_1, x_2) \in \mathcal{X}_0 \times \mathcal{X}_1\times \mathcal{X}_2$ except for one triple $(x_0^\star, x_1^\star, x_2^\star)$ that one can freely choose, define the following continuous real-valued functions:
\begin{equation}
g_{(x_0, x_1, x_2)}\colon p\in\mathcal{P} \mapsto p(x_0,x_1,x_2).
\end{equation}
Also define the continuous real-valued function $g_0$ as on top of the next page, see \eqref{eq:g0}.
\begin{figure*}[th!]
\begin{IEEEeqnarray}{rCl}\label{eq:g0}
\lefteqn{g_0\colon p\in\mathcal{P} \mapsto}\nonumber \\
& &- \sum_{(x_0', x_2')\in\mathcal{X}_0\times \mathcal{X}_2} \Bigg(\sum_{x_1'\in\mathcal{X}_1}p(x_0', x_1', x_2')\Bigg)\log \Bigg( \sum_{x_1'\in\mathcal{X}_1}p(x_0', x_1', x_2')\Bigg)\nonumber \\
&& +   \sum_{(x_2',y')\in\mathcal{X}_2\times \mathcal{Y}} \Bigg(\sum_{(x_0',x_1')\in\mathcal{X}_0\times \mathcal{X}_1} p(x_0', x_1', x_2')\Gamma(y'|x_0', x_1', x_2')\Bigg)\log \Bigg(\sum_{(x_0',x_1')\in\mathcal{X}_0\times\mathcal{X}_1} p(x_0', x_1', x_2')\Gamma(y'|x_0', x_1', x_2')\Bigg).
\end{IEEEeqnarray}
\hrule
\end{figure*}

Now, fix a $5-$uple $( V, X_0, X_1, X_2,Y)$ satisfying the conditions in the theorem, and where $V$ is allowed to be over any desired  alphabet $\mathcal{V}$ which can even be of unbounded cardinality. Let $Q_{X_0 X_1 X_2}$ denote the joint law of $(X_0, X_1, X_2)$ and $F_V(\cdot)$ the cumulative distribution function of $V$. For each $v\in\mathcal{V}$, let $p_{\mathbf{X}|V=v}(\cdot, \cdot, \cdot)\in \mathcal{P}$ denote the conditional law of the tuple $(X_0, X_1, X_2)$ given $V=v$. 

For any tuple $(x_0, x_1, x_2) \in \mathcal{X}_0 \times \mathcal{X}_1\times \mathcal{X}_2$ for which the function $g_{(x_0,x_1, x_2)}$ is defined, we have
\begin{equation}\label{eq:law}
\int_{\mathcal{V}} g_{(x_0,x_1, x_2)}(p_{\mathbf{X}|V=v}) \d F_V(v) = Q_{X_0X_1X_2}(x_0,x_1,x_2).
\end{equation}
Moreover, 
\begin{IEEEeqnarray}{rCl}\label{eq:gg}
\int_{\mathcal{V}} g_{0}(p_{\mathbf{X}|V=v})  \d  F_V(v) &=& H(X_0,X_2| V)-H(Y,X_2|V) \nonumber \\
& = &  H(X_0| V,X_2)-H(Y|V,X_2).\IEEEeqnarraynumspace
\end{IEEEeqnarray}

By the Support Lemma, \cite[Appendix~C]{elgamalkim-2011}, there exists a set $\tilde{\mathcal{V}}$ satisfying \eqref{eq:cardinality_bound}, a probability mass function ${Q}_{\tilde V}(\cdot)$ over  $\tilde{\mathcal{V}}$ , and $|\tilde{\mathcal{V}}|$ conditional probability distributions $\{ p_{v}\in\mathcal{P}\}_{v\in\tilde{\mathcal{V}}}$ such that for any tuple $(x_0, x_1, x_2) \in \mathcal{X}_0 \times \mathcal{X}_1\times \mathcal{X}_2$ for which the function $g_{(x_0,x_1, x_2)}$ is defined, 
\begin{align}\label{eq:keep_law}
&\int_{\mathcal{V}} g_{(x_0,x_1, x_2)}(p_{\mathbf{X}|V=v})  \d  F_V(v) \nonumber \\
&\phantom{==========} =  \sum_{v\in\tilde{\mathcal{V}}} g_{(x_0,x_1, x_2)}(p_v) {Q}_{\tilde V}(v)\IEEEeqnarraynumspace
\end{align} 
and 
\begin{equation}\label{eq:keepg0}
\int_{\mathcal{V}} g_{0}(p_{\mathbf{X}|V=v})  \d  F_V(v) = \sum_{v\in\tilde{\mathcal{V}}} g_0(p_v) {Q}_{\tilde V}(v). 
\end{equation}

Define the $5-$uple $(\widetilde V, \widetilde{X}_0, \widetilde X_1, \widetilde X_2, \widetilde Y)$ to be of law 
\begin{equation}\label{eq:deftilde}
\widetilde{Q}_V(v) \cdot p_v(x_0, x_1, x_2) \Gamma(y|x_0,x_1,x_2).\;
\end{equation}
By  definition~\eqref{eq:deftilde}, by \eqref{eq:law} and by \eqref{eq:keep_law},  the tuple $(\widetilde X_0, \widetilde X_1, \widetilde X_2, \widetilde Y)$ has the same law as the original tuple $( X_0,  X_1,  X_2,  Y)$:
\begin{IEEEeqnarray}{rCl}\label{eq:keep_joint_law}
\lefteqn{\Pr\big[{\widetilde X_0=x_0,\widetilde X_1=x_1,\widetilde X_2=x_2, \widetilde Y=y}\big] }\nonumber \\
& =& Q_{X_0, X_1, X_2}(x_0,x_1,x_2) \cdot \Gamma(y|x_0,x_1,x_2).\IEEEeqnarraynumspace
\end{IEEEeqnarray} 
Moreover, by \eqref{eq:gg} and \eqref{eq:keepg0}, and Definition~\eqref{eq:g0}, the relevant mutual informations are also preserved:
\begin{equation}
H(\widetilde{X}_0|\widetilde V, \widetilde{X}_2) - H(\widetilde Y |\widetilde V, \widetilde{X}_2) = H(X_0 |V, X_2) - H(Y|V, X_2),
\end{equation}
and as a consequence, by \eqref{eq:keep_joint_law},
\begin{equation}
 I(\widetilde Y ;\widetilde V| \widetilde{X}_2) -I(\widetilde{X}_0;\widetilde V| \widetilde{X}_2) =  I( Y ; V|X_2) -I({X}_0; V|X_2).
\end{equation}
This concludes the proof.

\section{Proof of Theorem~\ref{thm1}}\label{sec:state-amp}
\subsection{Achievability} 
Consider a joint distribution $Q_{UVX_0X_1X_2Y}\in\Delta(\mathcal{U}\times \mathcal{V}\times \mathcal{X}_0\times\mathcal{X}_1\times\mathcal{X}_2\times\mathcal{Y} )$ and a decoding function $g\colon \set{U}\times  \set{Y}\to \set{X}_2$ that satisfy Conditions~\eqref{eq:distor_causal}--\eqref{eq:cons21} in the theorem.

Fix small $\epsilon > \tilde \epsilon > \epsilon_3 > \epsilon_2 > \epsilon_1 >0$, and pick positive rates $R, \tilde{R}$ in a way that we specify later on.  

\textit{Codebooks generation:}
Split the blocklength $T$ into $B$ blocks each of length $n\eqdef\left\lfloor T/B\right \rfloor$. For each block $b\in\{1,\ldots, B\}$ randomly generate a codebook $\set{C}_{U}^{(b)}$ containing the $n$-length codewords $\{u^{(b)}(1),\ldots, u^{(b)}(\lfloor 2^{nR}\rfloor)\}$ and a codebook $\set{C}_{V}^{(b)}$ containing the $n$-length codewords $\big\{v^{(b)}(1,1),\ldots, v^{(b)}\big(\lfloor 2^{nR}\rfloor,\lfloor 2^{n\tilde{R}}\rfloor\big)\big\}$.  All entries of all codewords of codebook $\set{C}_{U}^{(b)}$  are drawn i.i.d. according to the marginal distribution $Q_{U}$. Independent thereof, all entries of all codewords of codebook $\set{C}_{V}^{(b)}$  are drawn i.i.d. according to the marginal distribution $Q_{V}$. 

\textit{Encoding:} Set $i_1=j_B=1$. For each block $b\in\{1,\ldots, B\}$, let $x_{0}^{(b)}$ denote the state sequence corresponding to block $b$.

For each block $b\in\{1,\ldots, B\}$, the encoder (Agent~$1$) looks for an index $i_{b}\in \{1,\ldots, \lfloor 2^{nR}\rfloor\}$ such that 
\begin{equation}
\Big(x_{0}^{(b)},u^{(b)}(i_{b})\Big) \in \set{T}_{\epsilon_1}^{(n)}(Q_{X_0U}).
\end{equation}
If there is more than one such index, it chooses the smallest among them, otherwise it declares an error.
For $b=1,\ldots, B-1$, set $j_b=i_{b+1}$. 

For each block $b \in\{1,\ldots, B\}$, the encoder looks for an index $\ell_{b}\in \{1,\ldots, \lfloor 2^{n\tilde{R}}\rfloor\}$ such that 
\begin{equation}
\Big(x_{0}^{(b)},u^{(b)}(i_{b}), v^{(b)}(j_b, \ell_b\big)\Big) \in \set{T}_{\epsilon_2}^{(n)}(Q_{X_0UV}). 
\end{equation}
If there is at least one such index, it picks one of them at random, otherwise it declares an error.
The encoder finally produces its $t'$-th input of block $b$, $x_{1,(b-1)n+t'}$, by applying  the conditional law $Q_{X_1|UVX_0}$ to the triple of symbols obtained by taking the $t'$-th components of the codewords $u^{(b)}(i_{b})$ and  $v^{(b)}(j_b, \ell_b)$ and the state vector $x_{0}^{(b)}$. 

\textit{Decoding:} Let $\hat{i}_1=1$.  
Fix $t\in\{1,\ldots,T\}$ and let $b$ denote the block to which time $t$ belongs to, i.e., $b= \lceil t/n \rceil$. Decoding at time $t\in\{1,\ldots, T\}$ depends on output $y_t$ and on the index $\hat{i}_{b}$ that---as we will see in a moment---the decoder (Agent~$2$) produced in a previous decoding step. Specifically, the decoder produces $x_{2,t}$  by applying the decoding function $g$ to the  $(t-(b-1)n)$-th component of codeword $u^{(b)}(\hat{i}_b)$ and to $y_t$. 



If $t$ is a multiple of $n$, i.e., we reached the end of a block, the decoder also looks for indices $(\hat{j}_{b}, \hat{\ell}_b)\in\big\{1,\ldots, \lfloor 2^{nR}\rfloor\big\} \times \big\{1,\ldots,\lfloor 2^{n\tilde{R}} \rfloor \big\}$ such that
\begin{equation}
\Big(u^{(b)}(\hat{i}_{b}), v^{(b)}(\hat{j}_b, \hat{\ell}_b), y^{(b)}\Big) \in \set{T}_{\epsilon_3}^{(n)}(Q_{UVY}). 
\end{equation}
If there is at least one such index, pick one of them at random. Otherwise declare an error.
Set $\hat{i}_{b+1}=\hat{j}_b$. 

\textit{Analysis:}
{We analyze the expected average distortion, where the expectation is taken with respect to the choice of the codebooks and the  random realizations of the state and the channel.}
Define for each block $b\in\{1,\ldots, B\}$, 
\begin{align}
  E_b \eqdef \Big\{(X_0^{(b)},X_1^{(b)},X_2^{(b)})\notin \set{T}_{\tilde{\epsilon}}^{(n)}(Q_{X_0X_1X_2})\Big\},
\end{align}
and 
\begin{align}
  E_{2:B} \eqdef \bigcup_{b=2}^B E_b.
\end{align}
We proceed to show that $\mathrm{P}(E_{2:B})$ can be made arbitrarily small for $n$ sufficiently large. 
We introduce the following events in each block $b\in\{1,\dots,B\}$.
\begin{align*}
  E_0^{(b)} &\triangleq \Big\{ \Big( X_{0}^{(b)},U^{(b)}(i)\Big) \notin \set{T}_{\epsilon_1}^{(n)}(Q_{X_0U})  \,\forall \, i \in \big\{1,\dots,\lfloor 2^{nR}\rfloor\big\}  \Big\}
  \\
  E_1^{(b)}&\triangleq \Big\{\Big(X_{0}^{(b)},U^{(b)}(i_{b}), V^{(b)}(j_b, \ell)) \notin \set{T}_{\epsilon_2}^{(n)}(Q_{X_0UV}) \nonumber \\
  &\phantom{==================} \forall \, \ell \in \big\{1,\dots,\lfloor2^{n\tilde{R}}\rfloor\big\}\Big\}\\
    E_2^{(b)}&\triangleq \Big\{\Big(X_{0}^{(b)},U^{(b)}(\hat{i}_{b}), V^{(b)}(j_b, \ell_b), X_{1}^{(b)}, Y^{(b)}\Big)\nonumber \\
     & \hspace{4cm} \notin \set{T}_{\epsilon_3}^{(n)}(Q_{{X_0UVX_1Y}}) \Big\}  \nonumber \\
  E_3^{(b)}&\triangleq \Big\{\Big( U^{(b)}(\hat{i}_{b}), V^{(b)}(j, \ell),  Y^{(b)}\Big) \in \set{T}_{\epsilon_3}^{(n)}(Q_{UVY})  \nonumber \\
  &\phantom{========} \textnormal{for some } j \in \{1,\dots,\lfloor2^{nR}\rfloor\} \backslash \{j_b\},  \nonumber \\
  & \hspace{4.6cm} \;\ell \in \big\{1,\dots,\lfloor2^{n\tilde{R}}\rfloor\big\}  \Big\}\nonumber \\
      E_4^{(b)}&\triangleq \Big\{\Big(X_{0}^{(b)},U^{(b)}(\hat{i}_{b}), V^{(b)}(j_b, \ell_b), X_{1}^{(b)}, Y^{(b)}, X_{2}^{(b)}\Big)\nonumber \\
     & \hspace{4cm} \notin \set{T}_{\tilde \epsilon}^{(n)}(Q_{{X_0UVX_1YX_2}}) \Big\}.
\end{align*} 
The probability $\mathrm{P}(E_{2:B})$ may be upper bounded as:\footnote{Here we also used the fact that event $E_2^{(b)c} \cap E_3^{(b)c}$ implies $\hat{i}_{b+1}=i_{b+1}$.}
\begin{align}
  \mathrm{P}(E_{2:B})& \leq \sum_{b=1}^{B} \bigg[ \mathrm{P}\big(E_0^{(b)} \big) + \mathrm{P} \big(E_1^{(b)}|E_0^{(b)c} \big)  \nonumber \\
  &\hspace{1cm}+ \mathrm{P} \big(E_2^{(b)}|E_1^{(b)c} \big) +\mathrm{P} \big(E_3^{(b)}|E_2^{(b)c} \big) \bigg] \nonumber \\
  & \hspace{1cm} +  \sum_{b=2}^{B} \mathrm{P} \big(E_4^{(b)}|E_2^{(b)c} \big)
 \end{align}
Throughout this paragraph, $\delta(\epsilon)$ stands for a function that tends to 0 as $\epsilon\to 0$.
\begin{itemize}
\item By the covering lemma \cite{elgamalkim-2011}, if $R > I_Q(X_0;U) + \delta(\epsilon_1)$, then for any $b \in \{1,2,\dots, B \}$: 
\begin{equation}
    \lim_{n\rightarrow\infty}\mathbb{E}\left(\mathrm{P}\big(E_0^{(b)}\big)\right)=0.
\end{equation}
%
\item By the covering lemma, if $\tilde{R} > I_Q(V;X_0 , U) + \delta(\epsilon_2)$, then for any $b \in \{1,2,\dots, B\}$: 
\begin{equation}
    \lim_{n\rightarrow\infty}\mathbb{E}\left(\mathrm{P}\big(E_1^{(b)}|E_0^{(b)c}\big)\right)=0.
\end{equation}

\item By the conditional typicality lemma \cite{elgamalkim-2011}, for any $b \in \{1,2,\dots, B\}$: 
\begin{equation}
    \lim_{n\rightarrow\infty}\mathbb{E}\left(\mathrm{P}\big(E_2^{(b)}|E_1^{(b)c} \big) \right)=0.
\end{equation}
\item By the packing lemma \cite{elgamalkim-2011}, if $R + \tilde{R} < I_Q(V;Y, U) - \delta(\epsilon_3)$, then for any $b \in \{1,2,\dots, B\}$: 
\begin{equation}
    \lim_{n\rightarrow\infty}\mathbb{E}\left(\mathrm{P}\big(E_3^{(b)}|E_2^{(b)c} \big) \right)=0.
\end{equation}


\end{itemize}

Whenever  $I_Q(X_0;U) < I_Q(V;Y,U)- I_Q(V;X_0,U)=I_Q(V;Y|U)- I_Q(V;X_0|U)$, and $\epsilon_1, \epsilon_2, \epsilon_3>0$ are sufficiently small,  it is possible to find rates $R, \tilde{R}>0$ such that 
\begin{subequations}\label{rate-constraints}
\begin{IEEEeqnarray}{rCl}
R&>& I_Q(X_0;U) +\delta(\epsilon_1)\\
\tilde{R} &>& I_Q(V;X_0,U)+\delta(\epsilon_2)\\
R+\tilde{R} &<& I_Q(V;Y, U)- \delta(\epsilon_3),
\end{IEEEeqnarray}
\end{subequations}

Thus, we conclude that $\mathrm{P}(E_{2:B})$ can be made arbitrarily small by choosing $n$ sufficiently large and $\epsilon_1, \epsilon_2, \epsilon_3>0$ sufficiently small. 
Define now 
\begin{equation}\label{eq:typicality}
E \eqdef \Big\{(X_0^{T},X_1^{T},X_2^{T})\notin \set{T}_{{\epsilon}}^{(T)}(Q_{X_0X_1X_2})\Big\}.
\end{equation}
 Since $\epsilon>\tilde \epsilon$, by choosing $B$ sufficiently large, the probability $\mathrm{P}( E )$ can be made as close to  $\mathrm{P}(E_{2:B})$ as one wishes. Thus, we conclude that when $n, B$ are sufficiently large and  $\epsilon_1, \epsilon_2, \epsilon_3>0$ are sufficiently small, it is possible to have 
\begin{equation}\label{eq:Ebound}
\mathrm{P}( E ) < \epsilon.
\end{equation}

Assume now that \eqref{eq:Ebound} holds. 
Under this assumption, we can bound the expected distortion (where the expectation is with respect to the choice of the codebooks and the channel realization) by
\begin{IEEEeqnarray}{rCl}
\lefteqn{
\mathbb{E}\left( \frac{1}{T}\sum_{t=1}^T \delta(X_{0,t}, X_{2,t})\right)}\qquad  \nonumber \\
 & = &\mathbb{E}\left( \frac{1}{T}\sum_{t=1}^T \delta(X_{0,t}, X_{2,t}) \Big| E\right)\mathrm{P}( E ) \nonumber \\
 & &  +   \mathbb{E}\left( \frac{1}{T}\sum_{t=1}^T \delta(X_{0,t}, X_{2,t}) \Big| E^c\right)\mathrm{P}( E^c ) \nonumber \\
 & \leq & d_{\max} \mathrm{P}( E ) + (D+ \epsilon d_{\max} ) \mathrm{P}( E^c )\nonumber \\
 & \leq & D+  2 \epsilon d_{\max} ,
 \end{IEEEeqnarray} 
 where the first equality follows by the total law of expectations; the first inequality because the distortion function is bounded and by the definition of the typical set $\set{T}_{{\epsilon}}^{(T)}(Q_{X_0X_1X_2})$; and the last inequality by \eqref{eq:Ebound} and because a probability cannot exceed 1.

{We see that the expected average distortion---where the expectation is taken with respect to the choice of the codebooks and the realizations of the state and the channel---can be made smaller than $D+2\epsilon d_{\max}$ when $T$ is sufficiently large.} {As a consequence, there must be at least one realization of all  codebooks  such that the expected average distortion is no larger than $D+2\epsilon d_{\max}$. Since $\epsilon>0$ can be chosen arbitrarily close to 0,} this concludes the achievability proof. 

By continuity  the above proof can be applied also when  $I_Q(X_0;U) \leq  I_Q(V;Y|U)- I_Q(V;X_0|U)$ holds with equality.

\subsection{Converse}
Let $D>0$ be an achievable distortion. Fix $\epsilon>0$ and $T$ sufficiently large. 
By definition, there exist coding and decoding functions $\{\sigma_t^{\text{NC}}\}_{t=1}^T$ and $\{\tau_t^{\text{C}}\}_{t=1}^T$ so that the sequences $X_1^T, X_2^T, Y^T$ induced by these functions and by the channel \eqref{eq:DMC2}, satisfy
\begin{equation}\label{eq:close8}
{\mathbb{E}\left( \frac{1}{T}\sum_{t=1}^T \delta(X_{0,t}, X_{2,t})\right) \leq D+ \epsilon. }
\end{equation}
 Let  $Z$ be a random variable that is uniformly distributed over $\{1,\ldots, T\}$ independent of $X_0^T, X_1^T, X_2^T,Y^T$, and define for each $t$, 
\begin{IEEEeqnarray}{rCl}\label{eq:VUt}
U_t & \eqdef & (Y^{t-1})\nonumber \\
V_t & \eqdef & (X_{0,t+1}^T).
\end{IEEEeqnarray}
Finally, let $U\eqdef (U_Z,Z)$, 
$V\eqdef (V_Z,Z)$, $X_0\eqdef X_{0,Z}$,  $X_1\eqdef X_{1,Z}$,  $X_{2}\eqdef X_{2,Z}$, and $Y\eqdef Y_Z$, and
denote the probability distribution of the tuple $(U,V,X_0,X_1,X_2,Y)$ by $Q_{UVX_0X_1X_2Y}$. By these definitions,
\begin{equation}\label{eq:average_over_time8}
\mathbb{E}\left(\delta(X_0,X_2)\right)= \mathbb{E}\left(\frac{1}{T}\sum_{t=1}^T \delta(X_{0,t}, X_{2,t})\right).
\end{equation}
and
\begin{IEEEeqnarray}{rCl}\label{eq:channel_factor8}
\lefteqn{
Q_{UVX_0X_1X_2Y}(u,v,x_0,x_1,x_2,y)} \nonumber \\
&  = & \rho_0(x_0)Q_{UVX_1X_2|X_0}(u,v,x_1,x_2|x_0) \Gamma(y|x_0,x_1). \end{IEEEeqnarray}
The Markov chain {$Y - (X_0,X_1) - (U,V, X_2)$} holds because by the i.i.d.-ness of the channel~\eqref{eq:DMC2}, $Y_t - (X_{0,t}, X_{1,t}) - (X_0^T, X_1^T,  X_2^T, Y^{t-1}, Y_{t+1}^T)$ forms a Markov chain for any $t\in\{1,\ldots, T\}$, and because the time-sharing random-variable $Z$ is independent of $(X_0^T, X_1^T, X_2^T, Y^T)$.

In the following, we prove that 
\begin{equation}\label{eq:information_constraint_to_prove8}
I_Q(X_0;U)\leq I_Q(V;Y|U)-I_Q(V;X_0|U),
\end{equation}
which combined with  \eqref{eq:close8}, \eqref{eq:average_over_time8}, and  \eqref{eq:channel_factor8}, by continuity, establishes the desired converse. 

\allowdisplaybreaks[4]
To prove \eqref{eq:information_constraint_to_prove8}, we first notice that on one hand,
\begin{IEEEeqnarray}{rCl}
\lefteqn{\frac{1}{T} I(X_0^T;Y^T)}\nonumber \\
 &=& \frac{1}{T} \sum_{t=1}^T I(X_{0,t};Y^T|X_{0,t+1}^T) \nonumber\\
& = &  \frac{1}{T} \sum_{t=1}^T I(X_{0,t};Y^T, X_{0,t+1}^T) \nonumber\\
& =& \frac{1}{T} \sum_{t=1}^T  \big[ I(X_{0,t};Y^t, X_{0,t+1}^T) + I(X_{0,t};Y_{t+1}^T| Y^t, X_{0,t+1}^T)\big]\nonumber\\
&= & \frac{1}{T} \sum_{t=1}^T  \big[  I(X_{0,t};Y_t,  V_t, U_{t}) + 
I(X_{0,t};Y_{t+1}^T| Y_t, U_t,V_t)\big],\label{eq:up18}\nonumber \\
& = &{ I(X_{0,Z}; Y_Z, V_Z, U_Z|Z) +  I(X_{0,Z};Y_{Z+1}^T| Y_Z, U_Z, V_Z, Z) }\nonumber \\
& = &{ I(X_0;Y,V,U) + I(X_{0,Z};Y_{Z+1}^T| Y_Z, U_Z, V_Z, Z).} \label{eq:up28}
\end{IEEEeqnarray}

On the other hand, we have that
\begin{IEEEeqnarray}{rCl}
\lefteqn{\frac{1}{T} I(X_0^T;Y^T)}\nonumber \\
& = & \frac{1}{T} \sum_{t=1}^T  I(X_{0,t};Y^T|X_{0,t+1}^T) \nonumber \\
& = & \frac{1}{T} \sum_{t=1}^T \Big[ I(X_{0,t};Y^T|X_{0,t+1}^T)  + I(Y^t;X_{0,t+1}^T) \nonumber \\
&& \phantom{=====================} - I(Y^{t-1};X_{0,t}^T) \Big]\nonumber \\
& =& \frac{1}{T} \sum_{t=1}^T \Big[ I(X_{0,t}^T;Y^t)  + I(X_{0,t}; Y_{t+1}^T|Y^{t}X_{0,t+1}^T) \nonumber \\
&& \phantom{=====================}  - I(Y^{t-1};X_{0,t}^T) \Big]\nonumber \\
&=& \frac{1}{T} \sum_{t=1}^T \Big[ I(X_{0,t}^T;Y_t |Y^{t-1})  + I(X_{0,t}; Y_{t+1}^T|Y^{t}X_{0,t+1}^T)  \Big]\nonumber \\
& =& \frac{1}{T} \sum_{t=1}^T \Big[ I(X_{0,t}, V_t; Y_t|U_t) +  I(X_{0,t};Y_{t+1}^T| Y_t,  U_t,V_t)  \Big] \nonumber \label{eq:up38}\nonumber \\
& =& I(X_{0,Z}, V_Z; Y_Z|U_{Z}, Z) +  I(X_{0,Z};Y_{Z+1}^T| Y_Z, U_Z, V_Z, Z) \nonumber \\
&=& I(X_0, V;Y|U) + I(X_{0,Z};Y_{Z+1}^T| Y_Z, U_Z, V_Z, Z). \label{eq:up48}
\end{IEEEeqnarray}
Combining \eqref{eq:up28} and \eqref{eq:up48}, we obtain
\begin{IEEEeqnarray}{rCl}
I_Q(X_0;Y, V, U) \leq  I_Q(X_0, V;Y|U),
\end{IEEEeqnarray}
which by  chain rule of mutual information is equivalent to 
\begin{IEEEeqnarray}{rCl}\label{eq:const8}
I_Q(X_0;U) &\leq& I_Q(X_0, V;Y|U)- I_Q(X_0; Y,V|U)\nonumber \\
&= & I_Q(V;Y|U) - I_Q(V;X_0|U). 
\end{IEEEeqnarray}

\bibliographystyle{IEEEtran}
\bibliography{bib-TIT-v9}

\clearpage

\end{document}